\documentclass[aps,
final,%
notitlepage,%
oneside,%
twocolumn,
nofootinbib,%
superscriptaddress,%
floatfix, prl]%
{revtex4}%changed from 4-2 so that PRL length checker would work
\pdfoutput=1

\usepackage{color}

\newcommand{\vertiii}[1]{{\left\vert\kern-0.25ex\left\vert\kern-0.25ex\left\vert #1 
		\right\vert\kern-0.25ex\right\vert\kern-0.25ex\right\vert}}

\newcommand{\eq}[1]{\begin{equation} \begin{aligned}#1\end{aligned} \end{equation}} %removed so that PRL length checker would work
\newcommand{\tr}{\mathrm{Tr}}

\newcommand{\ketbra}[2]{| #1 \rangle \langle #2 |}

\newcommand{\sugg}[1]{#1}

\newcommand{\trans}{T}
\newcommand{\mix}{\mathcal{M}}
\newcommand{\pur}{\mathcal{P}}
\newcommand{\entr}{\mathcal{S}}
\newcommand{\conc}{\mathcal{G}}
\newcommand{\cent}{\lambda}
\newcommand{\eu}{\text{e}}

\newcommand{\stkout}[1]{\ifmmode\text{\sout{\ensuremath{#1}}}\else\sout{#1}\fi}
\newcommand{\op}{F}
\newcommand{\opp}{G}

\usepackage{graphicx}% Include figure files
\usepackage{dcolumn}% Align table columns on decimal point
\usepackage{bm}% bold math
\usepackage{color}
\usepackage{soul}
\usepackage{amsmath}
\usepackage[dvipsnames]{xcolor}
\usepackage{mathtools}
\usepackage{dsfont}
\usepackage{amssymb}
\usepackage[utf8]{inputenc}
\usepackage[normalem]{ulem}

\usepackage{amsthm}
\theoremstyle{plain}
\newtheorem{thm}{\protect\theoremname}
  \theoremstyle{plain}
  \newtheorem{lem}[thm]{\protect\lemmaname}
    \newtheorem{proposition}[thm]{\protect\propositionname}
  \theoremstyle{remark}

\usepackage{babel}
  \providecommand{\claimname}{Claim}
  \providecommand{\corollaryname}{Corollary}
  \providecommand{\lemmaname}{Lemma}
  \providecommand{\factname}{Fact}
  \providecommand{\propositionname}{Proposition}
\providecommand{\theoremname}{Theorem}

\usepackage{hyperref}
\hypersetup{
	colorlinks=true,
	linkcolor=ForestGreen,
	filecolor=TealBlue,      
	urlcolor=Orchid,
	citecolor=Peach
}
\usepackage{pifont}
\usepackage{amssymb}

\usepackage{titlesec}

\begin{document}
		
	\title{Entanglement, loss, and quantumness: When balanced beam splitters are best}

\author{Noah Lupu-Gladstein}
\affiliation{National Research Council of Canada, 100 Sussex Drive, Ottawa, Ontario K1N 5A2, Canada}
\affiliation{Department of Physics, University of Ottawa, 25 Templeton Street, Ottawa, Ontario, K1N 6N5 Canada}
\author{Anaelle Hertz}
\affiliation{National Research Council of Canada, 100 Sussex Drive, Ottawa, Ontario K1N 5A2, Canada}
\author{Khabat Heshami}
\affiliation{National Research Council of Canada, 100 Sussex Drive, Ottawa, Ontario K1N 5A2, Canada}
\affiliation{Department of Physics, University of Ottawa, 25 Templeton Street, Ottawa, Ontario, K1N 6N5 Canada}
\affiliation{Institute for Quantum Science and Technology, Department of Physics and Astronomy, University of Calgary, Alberta T2N 1N4, Canada}
\author{Aaron Z. Goldberg}
\affiliation{National Research Council of Canada, 100 Sussex Drive, Ottawa, Ontario K1N 5A2, Canada}

	\setcounter{tocdepth}{1}

	\begin{abstract}
Quantum optics routinely uses beam splitters to generate entanglement, including in pioneering experiments conducted by Hanbury-Brown and Twiss and Hong, Ou, and Mandel. The quantum interference at beam splitters lies at the heart of what makes boson sampling hard to emulate by classical computers and is a vital component of quantum computation with light. Yet, despite overwhelming positive evidence, the conjecture that beam splitters with equal reflection and transmission probabilities generate the most entanglement for any state interfered with the vacuum  has remained unproven for almost two decades [Asb\'oth \textit{et al.}, Phys. Rev. Lett. \textbf{94}, 173602 (2005)]. We prove this conjecture for ubiquitous entanglement monotones including mixed-state generalizations of entanglement entropy and purity by uncovering monotonicity and convexity with respect to photon loss for these monotones. \sugg{At the same time, we highlight an infinite class of lesser-used monotones for which the conjecture fails.} Because beam splitters are so fundamental, our results yield numerous corollaries for quantum optics, including proof of a recent conjecture for the evolution of a measure of quantumness through loss \sugg{and a more efficient computational strategy for optimizing entanglement generation over linear optics. These results justify the value of seeking mathematical rigour behind commonly accepted facts and the danger of trusting them unconditionally.} 	\end{abstract}

 	\maketitle

Beam splitters are foundational to entanglement generation~\cite{HongOuMandel1987}, physical and mathematical models for optical loss and detector inefficiencies~\cite{YuenShapiro1978,YuenShapiro1980,CavesCrouch1987,Jeffersetal1993}, mathematical descriptions of mode transformations in classical and quantum optics~\cite{BornWolf1999,SalehTeich2007,MandelWolf1995}, interferometry~\cite{MichelsonMorley1887,LIGO2011}, and measuring wave-particle duality of light~\cite{Grangier_1986}. Beam splitters convert nonclassicality\sugg{---the failure of states to be modeled as probabilistic mixtures of coherent states~\cite{Sudarshan1963,Glauber1963,BachLuxmannEllinghaus1986,Sperling2016}---} into entanglement between modes.
Generating entanglement from nonclassicality is essential to the classical intractability of boson sampling~\cite{TroyanskyTishby1996,AaronsonArkhipov2013,Zhongetal2021,Arrazolaetal2021} and photonic quantum computation~\cite{Knilletal2001,BrowneRudolph2005,Koketal2007,MariEisert2012,TakedaFurusawa2019,SlussarenkoPryde2019,Flaminietal2019}.

Entanglement is a hallmark of quantum theory~\cite{EPR1935,Bell1964,Aspectetal1981} and crucial to many of its practical advantages~\cite{Bennettetal1993,Bouwmeesteretal1997,JozsaLinden2003,Mitchelletal2004,PezzeSmerzi2009}, making entanglement's manipulation, generation, and characterization essential.
The entangling properties of beam splitters have been studied from many perspectives %\cite{Tanetal1991,Sanders1992,HuangAgarwal1994,Paris1999,Wolfetal2003,Asbothetal2005,Ivanetal2006arxiv,UshaDevietal2006,Tahiraetal2009,Springeretal2009,Lietal2010,Ivanetal2011,Pianietal2011,Ivanetal2012generation,AbdelKhaleketal2012,Jiangetal2013,Berradaetal2013,Killoranetal2014,VogelSperling2014,Monteiroetal2015,Geetal2015,Miranowiczetal2015,Streltsovetal2015,Brunellietal2015,Arkhipovetal2016,RohithSudheesh2016,GholipourShahandeh2016,Maetal2016,Killoranetal2016,BoseKumar2017,Parketal2017,GoldbergJames2018,Kwonetal2019,Tasginetal2020,Fuetal2020,Tserkisetal2020,GoldbergHeshami2021,Ataman2022,Akellaetal2022,Lietal2023,Simonettietal2023,DeepakChatterjee2023,Steinhoff2024,SerranoEnsastigaetal2024arxiv,Lietal2024arxiv} 
\cite{Tanetal1991,HuangAgarwal1994,Paris1999,Wolfetal2003,Asbothetal2005,Tahiraetal2009,Springeretal2009,Pianietal2011,Ivanetal2012,Jiangetal2013,Killoranetal2014,Miranowiczetal2015,Maetal2016,BoseKumar2017,Tserkisetal2020,Steinhoff2024}.
They imply an important result for noisy quantum channels~\cite{Marietal2014} and are crucial for distributing, measuring, and distilling entanglement in quantum networks~\cite{Takahashietal2010,Sangouardetal2011,Loetal2012,Lucamarinietal2018}. 
%\red{and imply an important result for noisy quantum channels~\cite{Marietal2014}}.
%Include the really recent \cite{SerranoEnsastigaetal2024arxiv}? But they do all operations other than beam splitters! They do entanglement caused by unitaries acting on SU(2) coherent states averaged over all such states, but say SU(2) unitaries don't make entanglement on such states
These results motivated a quantification of \sugg{quantum-}optical nonclassicality called entanglement potential: the amount of entanglement a single-mode state can generate at a beam splitter when interfered with the vacuum \cite{Asbothetal2005}. Entanglement potential has been measured experimentally~\cite{Zavattaetal2007,Parigietal2007manipulating,Podhoraetal2020,Rohetal2023,Kadlecetal2024}. When entanglement potential was defined it was anecdotally observed that the most entanglement for any input state seemed to  always be generated by perfectly balanced beam splitters with equal transmission and reflection probabilities of 50\%~\cite{Wolfetal2003,Asbothetal2005,Tahiraetal2009,Brunellietal2015,Geetal2015,Monteiroetal2015,Wangetal2016,Arkhipovetal2016}. 

The 50\% condition appears in a number of other places, including distilling two-qubit entanglement~\cite{Sunetal2022,Nosratietal2024}, the switching of degradability properties of a channel~\cite{Leviantetal2022}, certain (but not all \cite{Silberhornetal2002,Pirandolaetal2020}) quantum key distribution schemes becoming useless after 50\% loss~\cite{GrosshansGrangier2002, Pirandolaetal2009, Takeokaetal2014, Pirandolaetal2017}, and balanced beam splitters often being ideal for interferometry~\cite{JarzynaDemkowiczDobrzanski2012,Geetal2020} and for quantum scissors~\cite{Peggetal1998}. States lose a witness of nonclassicality, the Wigner negativity, when they suffer at least 50\% loss~\cite{RadimFilip}, which touches upon the classical simulability of boson sampling under sufficient loss~\cite{RahimiKesharietal2016,OszmaniecBrod2018,Qietal2020}.

Despite the myriad of instances where 50\% loss achieves special importance, there has been no proof that balanced beam splitters generate the most entanglement for all input states~\cite{Asboth2024}\sugg{; therefore, despite its abundant use, the definition of entanglement potential has not yet been rigorously justified}. The optimality for Gaussian states is straightforward when measuring entanglement with a number of quantifiers~\cite{Wolfetal2003,Tahiraetal2009,Goldbergetal2023}. Proofs for Fock states can be gleaned from the statistics literature for entanglement measured by purity or entropy~\cite{SheppOlkin1981,GavreaIvan2014,Nikolov2014,HillionJohnson2017,Rasa2018,Rasa2019,HillionJohnson2019,Alzer2020}. Majorization has only offered results for comparing certain ranges of transmission probabilities~\cite{Gagatsosetal2013,vanHerstraetenetal2023}. Entropy power inequalities have been derived between the inputs and outputs of beam splitters~\cite{DePalmaetal2014,Shinetal2023}, but not between different beam splitters. We prove balanced beam splitters are best for all states and three quantifiers of entanglement via Theorems~\ref{thm:entr}, \ref{thm:pur}, and \ref{thm:conc}.

\begin{proposition}
For any input state interfered with vacuum on a beam splitter, the 50/50 beam splitter generates the most entanglement according to the entanglement monotones obtained via the convex-roof extension~\cite{BengtssonZyczkowski2006,Horodeckietal2009} (Eq.~\eqref{eq:convex roof extension}) of the following functions
\begin{enumerate}
    \item Entanglement entropy (Eq.~\eqref{eq:entropy})
    \item Mixedness (Eq.~\eqref{eq:mixedness})
    \item $G$-concurrence (Eq.~\eqref{eq:gconcurrence})
\end{enumerate}
\end{proposition}

% \sugg{We prove, for all initial states, that balanced (50:50) beam splitters generate the most entanglement as quantified by entanglement entropy (Theorem~\ref{thm:entr} and Eq.~\eqref{eq:derivative entropy two relative}), linear entropy (Theorem~\ref{thm:pur} and Eq.~\eqref{eq:purity positive polynomial lambda}), $G$-concurrence  (Theorem~\ref{thm:conc} and Eq.~\eqref{eq:conc log convex}) and convex-roof extensions thereof  (Eq.~\eqref{eq:convex roof extension}).}
% \st{We quantify entanglement in terms of ubiquitously used entanglement monotones and show that balanced beam splitters generate the most entanglement. }
\sugg{Beyond these three specific monotones, we show by counterexample that balanced beam splitters are suboptimal according to an infinite class of less commonly employed monotones~\cite{HillionJohnson2017}. Although we are not the first to discover counterexamples, we highlight them to emphasize the richness of linear optical entanglement and the subtlety of the original conjecture. }

\textit{Definitions (beam splitters, loss, entanglement monotones)}---A beam splitter mixes two modes via complex transmission and reflection coefficients. Without loss of generality, we consider a one-parameter family given by the real transmission probability $\trans \in [0, 1]$. Mathematically, a bosonic mode annihilated by $a$ impinging on a beam splitter with another bosonic mode ($b$) transforms with unitary $B(T)=\eu^{(a b^\dagger-a^\dagger b) \arccos\sqrt{\trans}}$ as~\cite{weedbrook} \begin{equation}
    a\to B(\trans)a B(\trans)^\dagger=\sqrt{\trans}a+\sqrt{1-\trans}b.
\end{equation}
When the $b$ mode begins in the vacuum and the other mode is any superposition of Fock states $|\psi\rangle\!=\!\sum_n \psi_n |n\rangle$,
 the resultant state is 
$B(\trans)|\psi,0\rangle=|\Psi_\trans\rangle=\sum_n \psi_n \sum_m \sqrt{\binom{n}{m}\trans^m (1-\trans)^{n-m}}|m,n-m\rangle$.

For nonclassical $|\psi\rangle$ and $\trans \in (0, 1)$, $|\Psi_\trans\rangle$ is entangled \cite{Aharonovetal1966,Kimetal2002,Xiangbin2002}. Entanglement presents mathematically in the Schmidt coefficients of $|\Psi_\trans\rangle$~\cite{Nielsen1999,Vidal2000}, which are the singular values of the Schmidt matrix $M_{mn}(\trans)=\langle m,n|\Psi_\trans\rangle$ and the square roots of the eigenvalues of the reduced state $\rho_\trans=\tr_b[|\Psi_\trans\rangle \langle \Psi_\trans|]$. $|\Psi_\trans\rangle$ is separable if and only if it has exactly $1$ non-zero Schmidt coefficient. Conversely, every Schur-concave function of the Schmidt coefficients is a pure-state entanglement monotone~\cite{Nielsen1999,Vidal2000}, where the latter are defined by never increasing under local operations and classical communication~\cite{BengtssonZyczkowski2006,Horodeckietal2009}.

When $\trans = 0$, the state becomes separable again. In fact, it becomes a swapped version of the initial state: $|\Psi_0\rangle = |0, \psi\rangle$. In general, the transformation $\trans\rightarrow 1-\trans$ is equivalent to swapping the modes in $|\Psi_\trans\rangle$ and transposing the Schmidt matrix. The Schmidt coefficients are invariant under this transformation, a symmetry that is exploited throughout this work. In particular, the symmetry implies that $\trans = 1/2$ is always a local optimum of any entanglement monotone based on Schmidt coefficients. The challenge of proving that 50/50 beam splitters generate the most entanglement amounts to showing that, for particular monotones, this local optimum is in fact the global maximum.

If the initial state of the $a$ mode is a mixed state $\rho_1$, the joint state after the beam splitter is mixed as well. Ignoring (tracing out) the $b$ mode results in $\rho_\trans = \tr_b[B(\trans)\rho_1\otimes|0\rangle\langle 0|B(\trans)^\dagger]$. The transformation $\rho_1 \rightarrow \rho_\trans = \mathcal{E}_\trans[\rho_1]$ is the standard model for optical loss~\cite{CavesCrouch1987,Jeffersetal1993,NielsenChuang2000,Leonhardt2003} and detector inefficiencies~\cite{YuenShapiro1978,YuenShapiro1980,Paul1982,Yurke1985,LeonhardtPaul1993} with transmission probability $\trans$. The loss channel $\mathcal{E}_\trans$ acts linearly on the input state and multiplicatively in $\trans$ in the sense that $\mathcal
{E}_{\trans_1}\circ \mathcal
{E}_{\trans_2}=\mathcal
{E}_{\trans_1\trans_2}$. In the Heisenberg picture, loss transforms the annihilation operator by $a \rightarrow \sqrt{\trans} a$~\cite{EkertKnight1991}. This mapping is the classical analog of damping an electric field amplitude $A$ via $A \rightarrow \sqrt{\trans} A$, as in the Beer-Lambert law~\cite{BornWolf1999,Jackson1999,Steck2016} and resembles the statistical binomial thinning operator~\cite{Jung2006}. Loss is generated by the Markovian dynamics~\cite{MandelWolf1995,ScullyZubairy1997,GardinerZoller2004,GerryKnight2005,supp} of zero-temperature amplitude damping~\cite{GardinerCollett1985}, itself a quantum generalization of the classical pure death process~\cite{Karlin2014}.

The link between beam splitters and loss arises because both processes model each photon as having an independent and identical probability $1\!-\!\trans$ to transfer into another mode and never return. A beam splitter accomplishes this by reflecting each photon into the same, well-controlled mode, whereas loss describes photons being absorbed or scattered into many different modes. Although the two processes are physically different, their mathematical effects on the original mode are equivalent. Modeling loss using fictitious beam splitters is known in quantum field theory as the thermo-field technique~\cite{Leonhardt2003,Umezawaetal1982,Takahashietal1996}.

%The input-output transformation for $a$ \cite{EkertKnight1991} is also the result of temporally integrating standard attenuation, absorption, or zero-temperature amplitude damping processes~\cite{GardinerCollett1985} with the Lindblad master equation
%$\dot{\rho}\propto a^\dagger a\rho+\rho a^\dagger a-2a\rho a^\dagger$~\cite{MandelWolf1995,ScullyZubairy1997,GardinerZoller2004,GerryKnight2005,supp}.

The entanglement monotones considered in this work all derive from R\'enyi entropies. The R\'enyi entropy $H_\alpha$ of order $\alpha$ is defined for any mixed state $\rho$ as $H_\alpha(\rho)= \lim_{\alpha' \rightarrow \alpha}\tfrac{1}{1-\alpha'}\log\tr[\rho^{\alpha'}]$. Every order $\alpha$ defines a corresponding pure-state entanglement monotone $H_\alpha(\tr_b[|\Psi\rangle \langle \Psi|)]$ and these pure-state monotones generate mixed-state monotones via the convex-roof extension~\cite{BengtssonZyczkowski2006,Horodeckietal2009} (Eq.~\eqref{eq:convex roof extension}). Our work focuses on three functions related to R\'enyi entropy.

The first is entanglement entropy~\cite{VonNeumann1932}
\begin{equation}\label{eq:entropy}
    \entr_\psi(\trans) = -\tr[\rho_\trans\log \rho_\trans] = H_{1}(\rho_\trans) ,
\end{equation}
which is important in fields including condensed matter physics~\cite{HuiHaldane2008}, quantum gravity~\cite{Headrick2010}, and uncertainty relations~\cite{ColesEUR}. It enjoys asymptotic continuity for interconverting between large numbers of entangled states~\cite{PopescuRohrlich1997,Vidal2000,Horodeckietal2009}. As well, many mixed-state entanglement monotones reduce to entanglement entropy for pure states, including distillable entanglement and relative entropy of entanglement (as in the original conjecture~\cite{Asbothetal2005}).
Mixedness $\mix_\psi(\trans)$ is defined in terms of the purity $\pur_\psi(\trans)$~\cite{Manfredi2000} as
\begin{equation}\label{eq:mixedness}
    1 - \mix_\psi(\trans)=\pur_\psi(\trans) = \tr[\rho_\trans^2]=\eu^{-H_2(\rho_\trans)}
\end{equation}
and is a close relative of R\'enyi 2-entropy, also known as linear or collision entropy. These functions provide more powerful entanglement witnesses than all Bell–Clauser-Horne-Shimony-Holt tests~\cite{Bovinoetal2005} and exhibit important convexity properties in terms of quantum states~\cite{LiebRuskai1973,Lindblad1974,Lindblad,Uhlmann1977}.
Finally, the $G$-concurrence~\cite{BarnumLinden2001,Sinoleckaetal2002,Gour2005} of a state with finite Schmidt rank $d$ is
\begin{equation}\label{eq:gconcurrence}
    \conc_\psi(\trans)=d\det [\rho_\trans]^{1/d} = \exp \left (\frac{\partial H_\alpha}{\partial \alpha}\Bigg|_{\alpha=0} \right ).
\end{equation}

\textit{Concavity of the von Neumann entropy}---We first prove the concavity of the von Neumann entropy of a state subject to loss:
\begin{thm}%[Concavity of entropy]
\label{thm:entr}
    The von Neumann entropy of any state subject to loss is a concave function of loss; ${\forall \rho_1 \geq0}, \ {\trans\in [0, 1]:}\  \partial ^2 H_1(\mathcal{E}_\trans[\rho_1])/\partial \trans^2\leq 0$.
    %$\partial ^2 \entr_\psi(\trans)/\partial \trans^2\leq 0\, \forall \trans\in(0,1)\, \forall \psi$.
    Consequently, any pure state interfered with the vacuum at a beam splitter generates the greatest entanglement entropy when the beam splitter has transmission $\trans=1/2$.
\end{thm}

The derivative of entropy with respect to $\trans$ is a difference between two relative entropies. When the initial state is pure, $\rho_1 = |\psi\rangle \langle \psi|$,
\begin{align}
    \frac{\partial \entr_\psi(\trans)}{\partial \trans}
    = &\langle a^\dagger a \rangle_\psi\! \left[
  \!  H \left ( \mathcal{E}_{1-\trans} \! \left [\frac{a  |\psi\rangle\langle\psi| a^\dagger}{\langle a^\dagger a \rangle_\psi} \right ] \!
    \Bigg\| \mathcal{E}_{1-\trans}  \left [|\psi\rangle\langle\psi| \right ] \right ) \right.\nonumber
    \\&\left.- H \left ( \mathcal{E}_{\trans} \left [\frac{a |\psi\rangle\langle\psi| a^\dagger}{\langle a^\dagger a \rangle_\psi} \right ]\Bigg\| \mathcal{E}_{\trans}[|\psi\rangle\langle\psi|] \right )
    \right].
  \label{eq:derivative entropy two relative}
\end{align}
In the above equation, which is derived in \cite{supp}, $H(X||Y)=\tr[X(\log X-\log Y)]$ is the relative entropy between two densities $X$ and $Y$. Physical processes cannot increase relative entropy~\cite{Lindblad}. Thus, the positive term is non-increasing as $\trans$ increases. The magnitude of the negative term moves in the opposite direction, but the negative sign means that both terms and their sum do not increase as $\trans$ increases. For initially mixed states $\rho_1\neq |\psi\rangle\langle \psi|$, the derivative can still be expressed as a difference of relative entropies by invoking a purification of the initial state and the same monotonicity logic carries over. The full derivative formula for initially mixed states and further mathematical details are presented in \cite{supp}.

Our proof features the photon subtracted state ${\sigma_T=a\rho_T a^\dagger/\mathrm{Tr}(a^\dagger a\rho_T)}$, which finds applications in quantum information processing~\cite{Ourjoumtsevetal2006,Parigietal2007,Walschaers2021}. The essential logic of the proof is that loss cannot make it easier to decide whether or not a photon has been subtracted from a state. 

 \setcounter{thm}{\getrefnumber{thm:entr}}

\textit{Concavity of mixedness/convexity of purity}---The task of proving the original conjecture is now complete. A curious reader might ask, is the conjecture true if we use a different measure of entanglement?  
Let us now turn to what is often the first entanglement monotone one learns about, mixedness.

\begin{thm}%[Concavity of mixedness/convexity of purity]
\label{thm:pur}
    The purity (mixedness) of a pure state subject to loss is a convex (concave) function of loss; $\forall |\psi\rangle, \trans \in [0, 1]$: $\partial ^2 \pur_\psi(\trans)/\partial \trans^2\geq 0$. Any pure state interfered with the vacuum at a beam splitter thus generates the greatest linear entanglement entropy when the beam splitter has transmission $\trans\!=\!1/2$.
\end{thm}

Our proof of this theorem rests on the fact that the purity of any state, and in fact the overlap between any two operators, can be measured using an operational procedure called the swap test~\cite{Bovinoetal2005,Islametal2015}. The swap test exploits the identity $\tr[F G] = \tr[(F \otimes G) S]$, where $F$ and $G$ are any trace-class operators and $S{=(-1)^{(a_1-a_2)^\dagger(a_1-a_2)/2}}$ is the swap operator. The utility of this formula is that it allows the overlap between two operators evolving in the Schr\"odinger picture to be expressed in terms of the evolution of $S$ in the Heisenberg picture. The Heisenberg evolution of $S$ under loss is detailed in \cite{supp} and leads to the identity
\begin{equation}\label{eq:purity positive polynomial lambda}
   \small{ \tr \left [\mathcal{E}_\trans[F] \mathcal{E}_\trans[G] \right ] \!= \!\tr \!\left [(F \otimes G) (\!1\!-\!2\trans)^{\frac{1}{2}(a_1-a_2)^\dagger (a_1-a_2)} \right].}\ 
\end{equation}
When $F = G = |\psi\rangle \langle \psi|$, the overlap equals the purity $\pur_\psi(\trans)$.
The operator {$(a_1-a_2)^\dagger(a_1-a_2)/2$} in Eq.~\eqref{eq:purity positive polynomial lambda} has non-negative integer eigenvalues and represents the number of photons in the dark port of a 50/50 beam splitter. As a result, $\pur_\psi(\trans)$ is a polynomial in ${1 - 2 \trans}$: $\pur_\psi(\trans) = \sum_{m \geq 0} p_m (1-2\trans)^m$. The coefficients $p_m$ represent the probability of measuring $m$ photons in the dark port of an interferometer in which $|\psi\rangle$ is interfered with a copy of itself on a 50/50 beam splitter. As a probability distribution, $p_m \geq 0$, thus $\pur_\psi(T)$ is monotonic and convex for $\trans \leq 1/2$. The $\trans \rightarrow 1 - \trans$ symmetry ensures that $\pur_\psi(\trans) = \pur_\psi(1-\trans)$. Thus convexity for $\trans \leq 1/2$ implies convexity for $\trans \geq 1/2$ as well.

% To be even more general, we inspect the overlap between any two positive, trace-class operators ${F_1}$ and ${G_1^\dagger}$ subject to loss, as measured by the Hilbert-Schmidt inner product $O_\trans=\tr[{F}_\trans {G}_\trans]$. One can show \cite{supp} that this overlap can be written as a polynomial in $\cent\equiv (1-2\trans)$ whose expansion coefficients are all nonnegative. In particular, by choosing $F_1=G_1=\rho_1$, this means that \begin{equation}\pur_\psi(\trans)=\sum_{m\geq 0}\lambda^m p_m
% \label{eq:purity positive polynomial lambda}
% \end{equation} for coefficients $p_m\geq 0$. 
% By symmetry about $\trans\!=\!1/2$ ($\lambda=0$), $p_m=0$ for all odd $m$. The second derivative is therefore comprised solely of positive terms $\frac{\partial^2 \pur_\psi(\trans)}{\partial \trans^2}=\sum_{l> 0}4(2l)(2l-1)\lambda^{2l-2}p_{2l} >0$ which proves the next theorem: 

% Symmetry about $\trans\!=\!1/2$ combined with the concavity properties proves the latter part of the theorem. As well, we learn that any two quantum states $\rho$ and $\tilde{\rho}$ have their Hilbert-Schmidt inner product decrease monotonically with $T<1/2$, showing how loss makes all quantum states less distinguishable.

\textit{Log-concavity of concurrence}---We continue to ask: does the optimality of balanced beam splitters hold for other entanglement monotones? One affirmative answer lies in the next theorem.

\begin{thm}%[Monotonicity of concurrence]
\label{thm:conc}
    The logarithm of the $G$-concurrence of a finite-dimensional pure state subject to loss 
    is a concave function of loss; 
    $\forall\,  |\{\psi_n\}_n|<\infty, \trans \in [0, 1]$: $\partial^2  \log\conc/\partial \trans^2\leq~0$. 
   % \begin{itemize}
   %      \item is everywhere nondecreasing for $T<1/2$
   %            \item is  everywhere  nonincreasing  for $T>1/2$
   %      \item  its logarithm is concave; 
   %      \item $\partial  \conc/\partial \trans\geq0\, \ \forall \,\trans\!\leq\! 1/2\, \forall |\{\psi_n\}|<\infty$ 
   %      \item $\partial  \conc/\partial \trans\leq0\, \ \forall \,\trans\!\geq 1/2\, \forall |\{\psi_n\}|<\infty$ 
   %      \item $\partial^2  \log\conc/\partial \trans^2\leq~0\forall\trans$.
   %  \end{itemize}
 % $T \substack{< \\ >}1/2$ and its logarithm is concave; $\partial  \conc/\partial \trans \substack{\geq \\ \leq}0\, \ \forall \,\trans \substack{\leq \\ \geq} 1/2\, \forall |\{\psi_n\}|<\infty$ and $\partial^2  \log\conc/\partial \trans^2\leq~0\forall\trans$. 
    Any pure state interfered with the vacuum at a beam splitter thus generates the greatest $G$-concurrence entanglement (entanglement capacity~\cite{Gour2005}) when the beam splitter has transmission $\trans\!=\!1/2$.
\end{thm}

The details of the proof appear in \cite{supp}, but the key is to notice that $\det[\rho_\trans]$ appearing in the definition of the  $G$-concurrence is the magnitude squared of the determinant of the Schmidt matrix. The Schmidt matrix is anti-diagonal, so the magnitude of its determinant is simply the product of the entries on the main anti-diagonal. These are completely determined by $\trans$, the maximum number of photons $N$ for which the initial state has non-zero amplitude, and the associated amplitude $|\psi_N|$. Specifically, 
\begin{equation}
\log\conc=\frac{N}{2}\log \left (\trans(1-\trans) \right ) + \log(|\psi_N|^2 (N+1)) + \frac{\log(C_N)}{N+1},
\label{eq:conc log convex}
\end{equation}
where $C_N = \Pi_{m=0}^N\binom{N}{m}$. This function is linear in $\log(\trans(1-\trans)$, which itself is concave in $\trans$. Thus $\log\conc$ is concave with a global maximum at $\trans=1/2$. By monotonicity of the logarithm, $\conc$ also has a global maximum at $\trans=1/2$.
This monotone only works for finite-dimensional states because determinants are not continuous when the dimension of a matrix changes~\cite{supp}.

This result shows that $\trans = 1/2$ maximizes $H_\alpha$ for infinitesimal $\alpha$. To first order in $\alpha$, $H_\alpha = H_0 + \alpha \log\conc + \mathcal{O}(\alpha^2)$. For $\trans \in (0, 1)$, $\rho_\trans$ is full rank and $H_0$ is independent of $\trans$. Thus, the first-order approximation of $H_\alpha$ is yet another entanglement monotone concave in $\trans$ and maximized at $\trans = 1/2$.

\textit{Mixed-state entanglement}---Pure states are excellent for idealizations and mathematical proofs, but the real world concerns mixed states. We thus extend our three above proofs to mixed-state entanglement measures.

For any pure-state entanglement monotone 
$E(\Psi)$, a mixed-state entanglement monotone can be defined by a convex-roof extension over all pure state decompositions~\cite{BengtssonZyczkowski2006,Horodeckietal2009}:
$\tilde{E}(\omega)=\inf_{\sum_i p_i |\Psi^i\rangle\langle\Psi^i| = \omega}\sum_i p_i E(\Psi^i)$. 
For example, the convex-roof extension of von Neumann entropy is known as the entanglement of formation~\cite{Bennettetal1996} and the extension of purity is called the generalized or I-concurrence~\cite{Rungtaetal2001}.

We argue that if $E(\Psi_T)$ is concave in $T$ for all initial pure states, then $\tilde{E}(\omega_T)$ is concave in $T$ for all initial states.  The mapping $\omega_T(\rho_1) = B(T) \rho_1 \otimes |0\rangle \langle 0| B(T)^\dagger$ is a one-to-one mapping from initial states to joint states after the beam splitter. Thus the infimum over joint state decompositions reduces to an infimum over input-state decompositions
\begin{equation}\label{eq:convex roof extension}
    \tilde{E}(\omega_T(\rho_1)) = \inf_{\sum_i p_i |\psi^i\rangle\langle\psi^i| = \rho_1} \sum_i p_i E \left (B(T) |\psi^i, 0\rangle \right ).
\end{equation}
% \begin{equation}\label{eq:convex roof extension}
%     \tilde{\mathcal{F}}_{\rho_1}=\inf_{\rho_1=\sum_i p_i |\psi^i\rangle\langle\psi^i|} \sum_i p_i \mathcal{F}_{\psi^i}(\trans)
% \end{equation}
A convex sum of concave functions and an infimum over concave functions are both concave themselves, thus $\tilde{E}(\omega_T(\rho_1))$ is concave in $T$. In fact, the same argument proves that $\tilde{E}$ inherits from $E$ any property preserved by convex sums and infimums, including symmetry about $\trans = 1/2$ and monotonicity. To successfully apply this argument, the property must hold for every initially pure state because the infimum scans over all pure state decompositions.

% This proves that the  entanglement generated by a beam splitter for any mixed state interfered with the vacuum, as measured by convex-roof extensions of $\entr_\psi$, $\mix_\psi$, and $\conc_\psi$, is maximized at $\trans\!=\!1/2$ (classical mixed states follow trivially), as well as concave for nonclassical mixed states for $\entr_\psi$, $\mix_\psi$, and $\log \conc_\psi$. 

\textit{Suboptimality for other monotones}---It is tempting to extrapolate from our results and conclude that balanced beam splitters are optimal according to all entanglement monotones. Since we have shown balanced is best for R\'enyi entropies of order $\alpha \in \{0, 1, 2\}$ and infinitesimal $\alpha$, it is natural to wonder whether balanced is best for all R\'enyi orders, especially considering how cherished entropy inequalities are throughout information theory~\cite{Shannon1948,Lindblad,BialynickiBirulaMycielski1975,ColesEUR,Hertz_2019,Rioul2011}.

We conjecture that balanced is best when interpolating to $\alpha \in [0, 2]$. \sugg{Of particular note is $\alpha = 1/2$, which corresponds to the log-negativity of entanglement (and the ``robustness of entanglement'' \cite{VidalTarrach1999}) when the input state is pure: $\parallel|\Psi_T\rangle\langle\Psi_T|^{\top_2}\parallel_1=\tr[\rho_T^{1/2}]^2$~\cite{VidalWerner2002,Calabreseetal2012}. Negativity-based measures enjoy computational tractability and can serve as an accurate proxy for more operationally meaningful measures like the the number of singlets one requires to make a given entangled state with restricted (``positive partial transpose'') operations~\cite{Audenaertetal2003,Horodeckietal2009}. Its connection to R\'enyi entropy immediately shows it achieves a local optimum at $\trans = 1/2$ but, as with the other $\alpha \in (0, 1) \cup (1, 2)$, the global optimality at $\trans = 1/2$ remains conjectural.}

On the other hand, balanced is not always best when extrapolating to $\alpha > 2$. 
Take the Fock state $|\psi\rangle=|6\rangle$ subject to loss. We plot in Fig.~\ref{fig:Renyi} its R\'enyi entropy versus $\trans$ for various values of $\alpha$, from which it is clear that higher-order entropies need not be concave nor maximized by balanced beam splitters. In fact, as soon as $\alpha > 2$, there are states for which $H_\alpha$ is not concave in $T$. In particular, for every $\alpha > 2$, there is a $T > 0$ at which $\partial^2 H_\alpha(\mathcal{E}_\trans[|1\rangle \langle1|])/\partial \trans^2$ changes sign~\cite{HillionJohnson2017}. Despite the loss of concavity when $\alpha > 2$, there are higher orders for which the entropy of Fock states is still optimized at $\trans \!=\! 1/2$. The value of $\alpha$ for which entanglement is no longer optimized at $\trans\!=\!1/2$ tends to be close to the Fock-state number but it can be both above or below that number; see~\cite{supp} for a list of values. 

The suboptimality of balanced beam splitters for higher order R\'enyi entropies is not merely a mathematical curiosity. The limit $\alpha \rightarrow \infty$ yields the min-entropy, which is the ultimate figure of merit for classical and quantum random number generation~\cite{Nisan1996,Konig2009,Ma2016}, so there is an operationally useful entanglement monotone for which balanced is not always best. Higher order R\'enyi entropies also appear in the replica trick~\cite{Holzheyetal1994,Headrick2010} used in quantum gravity and statistical physics that finds a series of R\'enyi entropies for integer-valued $\alpha\geq 2$ and applies an analytic continuation to find the von Neumann entropy. Somehow the peculiarities of higher-order entropies~\cite{BenBassatRaviv1978,KimSanders2010,Kwonetal2020} conspire to make the von Neumann entropy concave in loss. This reinforces that a specialized calculation is required for each entanglement monotone.

\begin{figure}
    \centering\includegraphics[width=\columnwidth]{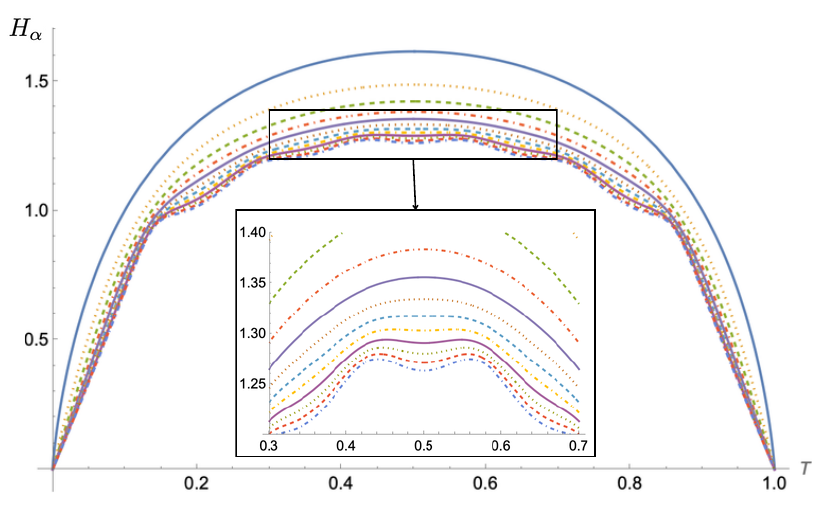}
    \caption{R\'enyi entropies from $\alpha=1$ (top, equal to von Neumann) to $\alpha=12$ (increasing downward) versus $\trans$ for a Fock state $|6\rangle$.}
    \label{fig:Renyi}
\end{figure}

\textit{Discussion}---Theorems \ref{thm:entr}, \ref{thm:pur}, and \ref{thm:conc} demonstrate three entanglement monotones for which balanced beam splitters generate the most entanglement between any state and the vacuum. This justifies the original definition of entanglement potential as the amount of entanglement generated by a state impinging on a balanced beam splitter~\cite{Asbothetal2005} as opposed to alternatives such as the entanglement generated by infinitesimally transmissive beam splitters~\cite{GoldbergHeshami2021} or averaged over beam splitter configurations~\cite{Goldbergetal2022sym}. Using the link between beam splitters and loss, we have shown that quantities such as entropy and mixedness of a pure state are concave with loss, no matter their dimensionality or Gaussianity.

Knowing properties of \textit{all} states at beam splitters is important to applications like teleportation-based schemes~\cite{NeergaardNielsenetal2013,Suetal2022} for quantum error correction and quantum communication~\cite{Suetal2022}. These schemes specifically rely on states’ interference with vacuum followed by photon number resolving detection. Similarly, schemes based on photon subtraction with beam splitters are key to inducing nonlocal correlations for entanglement distillation~\cite{Takahashietal2010}. Our proof paves the way for extending these schemes to arbitrary resource states by dramatically reducing the effort required to optimize over resource states and linear optics.

Consider finding the best way to convert single-mode nonclassicality into bipartite entanglement. \textit{A priori}, one must weigh all possible resource states with all possible linear optical networks. Thankfully, our results promise (for at least three popular entanglement monotones) that the best possible linear optical network is the humble balanced beam splitter~\cite{Asbothetal2005}. Furthermore, many entanglement monotones are convex over states~\cite{BengtssonZyczkowski2006,Horodeckietal2009}, ensuring the optimal resource state will be pure. This still holds even subject to linear constraints, such as requiring the resource state to have its average photon number not exceed an energy budget. 

% Our results also open the possibility of studying \textit{multi}partite entanglement in multimode systems, a topic for another day.

The connection between beam splitters and loss provides \sugg{further} applications for our results~\cite{CompanionLong}. The results on purity prove a recent conjecture of ours that a witness of nonclassicality called the quadrature coherence scale~\cite{Debievre,Hertz}
stops certifying nonclassicality at 50\% loss~\cite{HGH}. In phase space, loss is represented as a Gaussian convolution and our results yield inequalities about these convolutions.  
%These inequalities can also be represented in terms of annihilation operators and cast as uncertainty relations~\cite{CompanionLong}. 
Our formula for the overlap between two operators under loss proves that the Hilbert-Schmidt inner product between any two states increases monotonically with loss past 50\%.

Our results on von Neumann entropy\sugg{, in turn,} show that the mutual information between any state and the vacuum is concave \sugg{in beam splitter transmission.} Mutual information accounts for both classical and quantum correlations, whereas the monotones we have discussed until this point only reflect quantum entanglement. In statistical mechanics, there is a direct link between the stability of thermal equilibrium and the concavity of entropy~\cite{Lieb1999}. Quantum thermodynamics posits that the best description of a pure state after an unknown interaction with the environment is the one that maximizes the von Neumann entropy. The concavity of entropy in loss suggests there might be physical systems where chaotic linear optical transformations drive states towards a predictable and stable thermal equilibrium.

% the von Neumann entropy of a state is equal to the number of accessible microstates of the rest of the (pure-state) universe compatible with the state, any linear optical transformation of pure single-mode radiation, should the properties of the other modes be forgotten and all be equally likely due to noise or other dynamics, is most likely to yield a maximal-entropy state, converging to $\trans=1/2$. Or, if you know a single convex property of the input state such as that it has a fixed average photon number and then subject it to noise, our entropic convexity arguments could be utilized.

The theorems presented here are simple and, yet, the corollaries are numerous and far-reaching, so we direct the interested reader to our companion paper~\cite{CompanionLong}. Understanding loss is crucial for mitigating its effects on quantum optical technologies. Similar mathematics abound in two-mode physical systems, including aberrated optical cavities and coupled superfluids~\cite{Gutierrezetal2023} that could benefit from our results. %\sugg{If one instead wishes to study negativity-based measures, the pure-state case requires a proof that the R\'enyi entropy of order $1/2$ is concave or has a global maximum at $\trans\!=\!1/2$, due to $\parallel|\Psi_T\rangle\langle\Psi_T|^{\top_2}\parallel_1=\tr[\rho_T^{1/2}]^2$. Our counterexamples are all for entropies of order $\alpha>2>1/2$ and, so, there may be hopes for a proof but, just as each of our proofs required a separate method, a proof for negativity would be unique.} 
\sugg{A complete classification of which entanglement measures are and are not optimized by balanced beam splitters would surely bring even more insight to all of these systems, but such a comprehensive picture remains elusive.}

\begin{acknowledgments}
    Discussions with Stephan De Bi\`evre and Nicol\'as Quesada are gratefully acknowledged.
    The NRC headquarters is located on the traditional unceded territory of the Algonquin Anishinaabe and Mohawk people. NLG acknowledges funding from the NSERC Discovery Grant and Alliance programs.
\end{acknowledgments}

\appendix

 \section{Proof of Theorem 2}
To prove the Theorem for concavity of von Neumann entropy and thus of entanglement entropy, as well as its optimum at $1/2$, we use the following Lemmas. They establish the symmetry about $\trans\!=\!1/2$ (Lemma~\ref{lem:sym}), then gradually massage the derivative of this entropy into the difference between two relative entropies (Lemmas~\ref{lem:entr der diff} and~\ref{lem:T vs 1-T}) so that we can use monotonicity of quantum relative entropies to prove our result. We then extend the concavity proof to the concavity with respect to loss of the von Neumann entropy of mixed input states, where the symmetry about $\trans\!=\! 1/2$ no longer holds and this no longer serves as a direct measure of entanglement.
\begin{lem}[Schmidt symmetry]\label{lem:sym}
The Schmidt matrix of a pure state subject to loss $M(T)$, with matrix elements $M_{m n}(T) = \langle m, n|\Psi_T\rangle$, obeys $M(1\!-\!T) = M(T)^\top$. 
\end{lem}
\begin{proof}
    For the Hermitian swap operator defined by $S |m, n\rangle = |n, m\rangle$, we have $S |\Psi_\trans\rangle=|\Psi_{1-\trans}\rangle$, which is like physically swapping the two output ports of the beam splitter. Therefore $\langle m, n| \Psi_{\trans}\rangle = \langle n, m|\Psi_{1-\trans}\rangle$.
\end{proof}

Transposing a matrix does not change its singular values, so a direct consequence of Lemma \ref{lem:sym} is that $|\Psi_T\rangle$ and $|\Psi_{1\!-\!T}\rangle$ have the same Schmidt coefficients.

The following two Lemmas reveal the concavity of entropy by together showing that the \textit{first} derivative of entropy is monotonically decreasing such that its \textit{second} derivative is never positive.
\begin{lem}
    [Derivative of entropy of entanglement]\label{lem:entr der diff} $\partial \entr_\psi/\partial \trans=\langle a^\dagger a\rangle_\psi\left(H(\tau_\trans||\rho_\trans)-H(\sigma_\trans||\rho_\trans)\right)$ for relative entropy $H(X||Y)=\tr[X(\log X-\log Y)]$ and normalized states $\sigma_\trans\propto a\rho_\trans a^\dagger$ and $\tau_\trans\propto \sqrt{\rho_\trans}a^\dagger a\sqrt{\rho_\trans} $.
    %compare length:
    %$\sigma_\trans\propto a\rho_\trans a^\dagger$ and $\tau_\trans\propto \sqrt{\rho_\trans}a^\dagger a\sqrt{\rho_\trans} $ with $\tr[\sigma_\trans]=\tr[\tau_\trans]=1$.
    %compare length:
    %$\sigma_\trans= a\rho_\trans a^\dagger/\tr[a^\dagger a\rho_\trans]$ and $\tau_\trans= \sqrt{\rho_\trans}a^\dagger a\sqrt{\rho_\trans} /\tr[a^\dagger a\rho_\trans]$.
\end{lem}
\begin{proof}
    Derivatives of lossy states are known from master equations for damped harmonic oscillators (sometimes with other loss parametrizations like $\eu^{-\gamma t}=\cos^2\theta=\eta=\trans$)~\cite{MandelWolf1995,ScullyZubairy1997,GardinerZoller2004,GerryKnight2005,WeidlichHaake1965,KarrleinGrabert1997}:
\begin{equation}
    \frac{\partial \rho}{\partial T}=\frac{1}{2T}\left( a^\dagger a\rho+\rho a^\dagger a-2a\rho a^\dagger\right).
\end{equation}
    The derivative of the von Neumann entropy of $\rho_T$ can thus be computed as
    \begin{equation}
        \begin{aligned}
            \frac{\partial \entr_\psi(\trans)}{\partial \trans}&=-\tr\left[\frac{\partial \rho_\trans}{\partial \trans}\log \rho_\trans\right]{+0}\\
    &=\frac{1}{2\trans}\tr[(2a\rho_\trans a^\dagger-a^\dagger a\rho_\trans-\rho_\trans a^\dagger a)\log\rho_\trans]\\
    =\langle a^\dagger &a \rangle_\psi \tr\left[\left(\frac{a\rho_\trans a^\dagger}{\tr[a^\dagger a\rho_\trans]}-\frac{\sqrt{\rho_\trans} a^\dagger a\sqrt{\rho_\trans}}{\tr[a^\dagger a\rho_\trans]}\right)\log\rho_\trans\right], 
        \end{aligned}
    \end{equation} where
    the third line uses the commutation of functions of $\rho_\trans$ and that $\tr[a^\dagger a\rho_\trans]=\trans\langle\psi|a^\dagger a|\psi\rangle$ because the average number of photons decreases by a factor of $\trans$ as expected from classical attenuation. Next, noticing that $\sigma_\trans=WW^\dagger$ and $\tau_\trans=W^\dagger W$ have the same eigenvalues such that $\tr[\sigma_\trans\log\sigma_\trans]=\tr[\tau_\trans\log\tau_\trans]$ and selecting $W=a\sqrt{\rho_\trans}/\sqrt{\tr[a^\dagger a \rho_\trans]}$, we add and subtract $\langle a^\dagger a\rangle_\psi\tr[\sigma_\trans\log\sigma_\trans]$ to complete the proof.
\end{proof}
We have the first derivative in terms of a difference between two relative entropies, but it is not immediately clear what their monotonicity properties are because we do not yet know whether $\sigma_\trans=\mathcal{E}_\trans(\sigma_1)$ and likewise for $\tau_\trans$; in fact, this only holds for the former. We thus transform the expression by further specifying the relationship between $\sigma_\trans$ and $\tau_\trans$.
\begin{lem}[Symmetry of relative entropy about loss]\label{lem:T vs 1-T}
    $H(\tau_\trans||\rho_\trans)=H(\sigma_{1-\trans}||\rho_{1-\trans})$.
\end{lem}
\begin{proof}
    First write $X=\sum_n K_n |\psi\rangle\langle n|$ for Kraus operators~\cite{Goldberg2024} 
    \eq{K_n(\trans)=\langle n|B(\trans)|0\rangle_b=\sqrt{\trans}^{a^\dagger a}(\sqrt{1-\trans}a)^n/\sqrt{n!}}such that $\rho_\trans=XX^\dagger$. Our notation occasionally suppresses $\trans$ dependence for brevity. 
    
        Then, note that any operator $X$ satisfies \eq{\sqrt{X X^{\dagger}} \log \left(X X^{\dagger}\right) \sqrt{X X^{\dagger}}=X \log \left(X^{\dagger} X\right) X^{\dagger}.}
        Indeed, using the polar decomposition any square matrix can be written as $X=PU$ where $P$ is positive semi-definite  and $U$ is unitary. 
        Then,
    $\sqrt{X X^{\dagger}} \log \left(X X^{\dagger}\right) \sqrt{X X^{\dagger}}=P\log(P^2)P$ and 
    $X \log \left(X^{\dagger} X\right) X^{\dagger}=PU\log(U^\dagger P^2U)U^\dagger P=PUU^\dagger\log(P^2)U U^\dagger P=P\log(P^2)P$. The logarithm is always defined on a positive, invertible matrix, even if $X$ is unbounded (which is not the case here), and we always have $\mathrm{Tr}(P^2)=1$ here.
     
     Thus, using the definitions of $\tau_\trans$ and $W$ introduced in the previous lemma, we have
     \begin{equation}
        \tr [\tau_\trans \log \rho_\trans] =\tr \left[\frac{X^{\dag}a^{\dag} a X }{\tr \left[a^{\dag} a X X^\dag \right]}\log \left(X^{\dag} X\right)\right]
    \end{equation} and $H_1[\tau_\trans]\!=\!H_1[\frac{X^{\dag}a^{\dag} a X }{\tr \left[a^{\dag} a X X^\dag \right]}]$. 
    
 Then, using the Baker–Campbell–Hausdorff identity and $[a^\dag a,a]=-a$ we prove that $aK_{n}=\sqrt{T}K_{n} a$ as follows\footnote{Note that this can also be proven by looking at matrix elements in the number basis.}
\begin{equation}
    \begin{aligned}
\sqrt{T}^{-a^\dag a} a \sqrt{T}^{a^\dag a} & =e^{-a^\dag a \ln \sqrt{T}} a e^{a^\dag a \ln \sqrt{T}}\\
& =a+[-a^\dag a \ln \sqrt{T}, a]\nonumber\\
&\ +\frac{1}{2!}[-a^\dag a \ln \sqrt{T},[-a^\dag a \ln \sqrt{T}, a]]+\ldots \\
%& =a+\ln \sqrt{T} a+\frac{1}{2!}[-N \ln \sqrt{T}, \ln \sqrt{T} a]+\ldots \\
%& =a+\ln \sqrt{T} a+\frac{(\ln \sqrt{T})^{2}}{2!} a+\cdots\\
&=a e^{\ln \sqrt{T}}=a \sqrt{T}\\
%\Rightarrow a \sqrt{T}^{a^\dag a} &  =\sqrt{T}^{a^\dag a+1} a \\
\Rightarrow a K_{n}&  =\sqrt{T} K_{n}(T) a .
\end{aligned}
\end{equation}
With this identity and $\sum_n K_n^\dagger K_n=\openone$, we find
    \begin{equation}
        \frac{X^\dag a^\dag a X}{\tr(X^\dag a^\dag a X)}=
        %\frac{\trans\sum_{m,n}K_m^\dag K_n\ketbra{m}{\psi}a^\dag a\ketbra{\psi}{n}}{\trans\sum_n K_n^\dag K_n \langle \psi|a^\dag a|\psi \rangle}
         \frac{T\sum_{m,n}\ketbra{m}{\psi}a^\dag K_m^\dag K_n a\ketbra{\psi}{n}}{T\sum_n  \langle \psi|a^\dag K_n^\dag K_n a|\psi \rangle}
         =Y^\dagger Y
    \end{equation}
    for $Y=\sum_n K_n a|\psi\rangle\langle n|/\sqrt{\langle a^\dagger a\rangle_\psi}$. This means that $H(\tau_\trans||\rho_\trans)=H(Y^\dagger Y||X^\dagger X)$. 
    
    Next, note that $X$ and $Y$ are equivalent to the Schmidt matrix representations of $B(T) |\psi, 0\rangle$ and $B(T) a |\psi, 0\rangle/\sqrt{\langle a^\dagger a  \rangle_\psi}$, that is:
    \begin{equation}
    \begin{aligned}
        \langle m|X|n\rangle&=\langle mn|B(T)|\psi,0\rangle ,\\ \langle m|Y|n\rangle&=\frac{\langle mn|B(T)a|\psi,0\rangle}{\sqrt{\langle a^\dagger a\rangle_\psi}}.
           \end{aligned}
    \end{equation}
    Therefore, Lemma \ref{lem:sym} reveals that we can transpose the matrix if we send $\trans\to1-\trans$
    \begin{equation}
        \begin{aligned}
            % H(\tau_\trans||\rho_\trans)&\!=\!H(Y(1\!-\!T)^*Y(1\!-\!T)^\top\!||X(1\!-\!T)^*X(1\!-\!T)^\top\!)\\
            H(\tau_\trans||\rho_\trans)&=H(Y_{1-T}^*Y_{1-T}^\top||X_{1-T}^*X_{1-T}^\top)\\
%H(\tau_\trans||\rho_\trans)&=H(Y(R)^*Y(R)^\top||X(R)^*X(R)^\top)\\
    &\!=\!H(Y_{1-T}^{}Y_{1-T}^\dagger||X_{1-T}^{}X_{1-T}^\dagger). 
        \end{aligned}
    \end{equation} The last line follows because entropy is real; transposition and complex conjugation are always taken with respect to the Fock basis throughout this work.
Finally, notice $X_{1-T}^{}X_{1-T}^\dagger\!=\!\rho_{1\!-\!T}$ and $Y_{1-T}^{}Y_{1-T}^\dagger\!=\!\sigma_{1\!-\!T}$.
\end{proof}

\setcounter{thm}{1}

We are now ready to prove the theorem in the main text:
\begin{thm}[Concavity of entropy]\label{thm:entr}
   The von Neumann entropy of any state subject to loss is a concave function of loss; ${\forall \rho_1 \geq0}, \ {\trans\in [0, 1]:}\  \partial ^2 H_1(\mathcal{E}_\trans[\rho_1])/\partial \trans^2\leq 0$.
    %$\partial ^2 \entr_\psi(\trans)/\partial \trans^2\leq 0\, \forall \trans\in(0,1)\, \forall \psi$.
    Consequently, any pure state interfered with the vacuum at a beam splitter generates the greatest entanglement entropy when the beam splitter has transmission $\trans=1/2$.
\end{thm}

\setcounter{thm}{\getrefnumber{lem:T vs 1-T}}

\begin{proof}
    By Lemmas~\ref{lem:entr der diff} and~\ref{lem:T vs 1-T}, $$\partial \entr_\psi(T)/\partial \trans=\langle a^\dagger a\rangle_\psi\left(H(\sigma_{1-\trans}||\rho_{1-\trans})-H(\sigma_{\trans}||\rho_{\trans})\right).$$ We will show that $H(\sigma_{\trans}||\rho_{\trans})$ increases monotonically with $\trans$; by symmetry, $H(\sigma_{1-\trans}||\rho_{1-\trans})$ must decrease with $\trans$ such that $\partial \entr_\psi/\partial \trans$ is nonincreasing with $\trans$.

    This all follows because $\sigma_\trans$ is the result of the loss channel acting on $\sigma_1$ that enacts $\rho_\trans=\mathcal{E}_\trans[\rho_1]$ and such a channel $\mathcal{E}_\trans$ is completely positive and trace preserving. Since loss is multiplicative, for $\trans\leq \trans^\prime$, $H(\sigma_{\trans}||\rho_{\trans})=H(\mathcal{E}_{\trans/\trans^\prime}[\sigma_{\trans^\prime}]||\mathcal{E}_{\trans/\trans^\prime}[\rho_{\trans^\prime}]|)\leq H(\sigma_{\trans^\prime}||\rho_{\trans^\prime})$, where the inequality holds by the monotonicity of relative entropy~\cite{Lindblad}. The relevant assertion is proven by
    \begin{equation}
        \begin{aligned}
            \sigma_\trans & =\frac{a \mathcal{E}_\trans[\rho_1] a^{\dag}}{\trans\ \tr \left(a^{\dag} \rho_1 a \right)}=\sum_n\frac{a  K_n\rho_1 K_n^\dagger a^{\dag}}{\trans\ \tr \left(a^{\dag} \rho_1 a \right)} \\
& =\sum_n K_n\frac{   a\rho_1 a^{\dag} }{\ \tr \left(a^{\dag} \rho_1 a \right)} K_n^\dagger=\mathcal{E}_\trans[\sigma_1] ,
%\nonumber 
        \end{aligned}
    \end{equation}again using $aK_n=\sqrt{\trans}K_n a$.\footnote{We note in passing that states of the form of $\sigma_\trans$ have gained prominence as ``photon-subtracted'' versions of $\rho_\trans$~\cite{Ourjoumtsevetal2006,Parigietal2007,Walschaers2021}.} Thus, $\partial ^2 \entr_\psi(\trans)/\partial \trans^2\leq 0\, \forall \trans\in(0,1)\, \forall \psi$. Concavity together with symmetry prove that the maximal entanglement is always generated at $\trans\!=\!1/2$ and that more entanglement is always generated the closer the parameter $\trans$ is to $1/2$. 
\end{proof}

As mentioned, the concavity of the entropy of states subject to loss is actually a general property of all input states. Our previous restriction to pure states was because it is for those states that this entropy corresponded to an entanglement measure. For mixed states, the appropriate entanglement measure is found in a different way that is detailed at the end of the main text. However, we still provide the proof of concavity on von Neumann entropy for mixed states $\rho_1$ as an independent result.

% \section{More element for the proof of Lemma 7 \red{Remove?}}

% \textbf{Result 3: }  $\sum_n \frac{(1-T)^n}{n!}(a_-^\dag)^n(-T)^{a_-^\dag a_-}(a_-)^n=\lambda^{a_-^\dag a_-}$

% \begin{proof}
% It is easy to see that $\sum_n \frac{(1-T)^n}{n!}(a_-^\dag)^n(-T)^{a_-^\dag a_-}(a_-)^n$ is diagonal in the number basis (of the mode $a_-$). 
% This means that the only non-zero elements are 
% \eqarray{&&\bra{m}\sum_n \frac{(1-T)^n}{n!}(c_-^\dag)^n(-T)^{\hat n_{c_-}}(c_-)^n\ket{m}
% \nonumber\\
% &=&
% \bra{m}\sum_n^m \frac{(1-T)^n m!}{n!(m-n)!}(-T)^{m-n}\ket{m}\nonumber\\
% &=&\sum_n^m \binom{m}{n}(1-T)^n(-T)^{m-n}\nonumber\\
% &=&\big((1-T)+(-T)\big)^m=(1-2T)^m=\lambda^m
% }
% and we can thus replace  $\sum_n \frac{(1-T)^n}{n!}(a_-^\dag)^n(-T)^{a_-^\dag a_-}(a_-)^n$ by $\lambda^{a_-^\dag a_-}$.

% \end{proof}

\section{ Concavity of von Neumann entropy for mixed states $\rho_1$.}

We extend the pure-state proof by writing
\begin{equation}
    \tilde{X}_\trans= \sum_i\sqrt{p_i} \sum_n K_n(\trans)|\psi^i\rangle\langle n|\otimes |0\rangle\langle \psi^{i\ast}|.
\end{equation} Here we have added an extra degree of freedom to take care of the mixedness of the initial, lossless state $\rho_1=\sum_i p_i |\psi^i\rangle\langle \psi^i|$, where the coefficients are valid probabilities $p_i>0$ and the eigenstates $\{|\psi^i\rangle\}$ are orthonormal without loss of generality, and we employ the orthonormal states $\{|\psi^{i\ast}\rangle\}$ whose coefficients are the complex conjugates of $\{|\psi^i\rangle\}$'s in the Fock basis. Considering the state that purifies $\rho_1$ to be ${|\Psi\rangle=\sum_i \sqrt{p_i}|\psi^i,\psi^i\rangle}$, we can write the above as ${\rho_1=\tr_b[|\Psi\rangle\langle \Psi|]}$ and \begin{equation}
    \begin{aligned}
        \tilde{X}_\trans^\top&
    =\sum_i\sqrt{p_i} \sum_n (K_n(\trans)|\psi^i\rangle\langle n|)^\top\otimes |\psi^i\rangle\langle 0|\\
    &=\sum_i\sqrt{p_i} \sum_n K_n(1-\trans)\otimes \openone|\psi^i,\psi^{i}\rangle\langle n, 0|\\
    &=\sum_n K_n(1-\trans)\otimes \openone|\Psi\rangle\langle n, 0|
    \end{aligned}
\end{equation} using the transpose of the Kraus operators proven for pure states. In this case, just like for pure states, \begin{equation}
    \tilde{X}_\trans\tilde{X}_\trans^\dagger=\mathcal{E}_\trans \otimes \openone [\rho_1\otimes|0\rangle\langle 0|]=\rho_\trans\otimes |0\rangle\langle 0|
\end{equation}from orthonormality of $\{|\psi^{i\ast}\rangle\}$ in the second degree of freedom and
\begin{equation}
    \begin{aligned}
        \left(\tilde{X}_\trans^\dagger\tilde{X}_\trans\right)^\ast
        &=\tilde{X}_\trans^\top\tilde{X}_\trans^{\top \dagger}\\
        &=
        \sum_n K_n(1-\trans)\otimes \openone|\Psi\rangle\langle \Psi|K_n(1-\trans)^\dagger\otimes \openone\\
        &=\mathcal{E}_{1-\trans} \otimes \openone [|\Psi\rangle\langle \Psi|].
    \end{aligned}
    \label{eq:asterisk X}
\end{equation}
Next, we adapt our pure-state proof to again use $\sqrt{X X^{\dagger}} \log \left(X X^{\dagger}\right) \sqrt{X X^{\dagger}}=X \log \left(X^{\dagger} X\right) X^{\dagger}$:
\begin{equation}
    \begin{aligned}
        \tr[a^\dagger a& \rho_\trans \log\rho_\trans]%=\tr[a^\dagger a \rho_\trans \log\rho_\trans]^\ast
        =\tr[\sqrt{\rho_\trans}a^\dagger a \sqrt{\rho_\trans} \log\rho_\trans]\\
        &=\tr[\sqrt{\rho_\trans \otimes |0\rangle\langle 0| }(a^\dagger a \otimes \openone)\sqrt{\rho_\trans \otimes |0\rangle\langle 0| }\\
        &\qquad\times\log(\rho_\trans\otimes |0\rangle\langle 0|)]\\
        &=\tr[\tilde{X}_\trans^\dagger (a^\dagger a\otimes \openone) \tilde{X}_\trans\log(\tilde{X}_\trans^\dagger\tilde{X}_\trans)].
    \end{aligned}
\end{equation}
Similarly defining
\begin{equation}
    \begin{aligned}
        \Tilde{Y}_{\trans}&=\frac{(a\otimes \openone)\tilde{X}_\trans}{\sqrt{\tr[\tilde{X}^\dagger (a\otimes \openone)^\dagger(a\otimes \openone)\tilde{X}_\trans]}}\\
        &=\frac{(a\otimes \openone)\tilde{X}_\trans}{\sqrt{T\tr[a^\dagger a \rho_1]}}
        \\
        %&=\left(\sum_n K_n \frac{a\otimes \openone |\Psi\rangle}{\sqrt{\langle \Psi| a^\dagger a\otimes \openone |\Psi\rangle}}\langle n,0|\right)^{\top_b},\\
        &=
        \sum_i\sqrt{p_i} \sum_n K_n(\trans)\frac{a |\psi^i\rangle}{\sqrt{\langle \Psi| a^\dagger a\otimes \openone |\Psi\rangle}}
        \langle n|\otimes |0\rangle\langle \psi^{i\ast}|,
    \end{aligned}
\end{equation} we find
\begin{equation}
    \begin{aligned}
        \tilde{Y}_{\trans}^\top=\sum_n \frac{K_n(1-\trans)a\otimes \openone |\Psi\rangle}{\sqrt{\langle \Psi| a^\dagger a\otimes \openone |\Psi\rangle}}
        \langle n,0|
    \end{aligned}
\end{equation} using the transposition identity
and thus
\begin{equation}
    \begin{aligned}
        \tr[a^\dagger a& \rho_\trans \log\rho_\trans]=\tr[a^\dagger a \rho_\trans \log\rho_\trans]^\ast
        \\
        &=\tr[a^\dagger a \rho_1]\tr[\Tilde{Y}_{\trans}^\dagger\Tilde{Y}_{\trans}\log (\tilde{X}_\trans^\dagger \tilde{X}_\trans]^\ast\\
        &=\tr[a^\dagger a \rho_1]\tr[\Tilde{\sigma}_{1-\trans}\log \Tilde{\rho}_{1-\trans}]
    \end{aligned}
\end{equation} for analogously defined $\Tilde{\rho}_{\trans}=\mathcal{E}_\trans\otimes\openone [|\Psi\rangle\langle \Psi|]$ and ${\Tilde{\sigma}_{\trans}=\mathcal{E}_\trans\otimes\openone [(a\otimes\openone)|\Psi\rangle\langle \Psi|(a^\dagger\otimes\openone)/\langle \Psi| a^\dagger a\otimes \openone |\Psi\rangle]}$. We thus have
\begin{equation}
    \frac{\partial H_1[\rho_\trans]}{\partial \trans}=\langle a^\dagger a\rangle_{\rho_1}\left(H(\Tilde{\sigma}_{1-\trans}||\Tilde{\rho}_{1-\trans})-H({\sigma}_{\trans}||{\rho}_{\trans})\right)
\end{equation}
and monotonicity of relative entropy is again enough to establish that this is positive monotonically increasing with $\trans$ such that the entropy is always concave. Notice that only one of the relative entropies above involves states $\Tilde{\rho}_\trans$ and $\Tilde{\sigma}_\trans$ in the expanded Hilbert space. This reflects the fact that the entropy of some initially mixed state is not symmetric about $\trans = 1/2$.

% We finally establish the promised piece, Eq.~\eqref{eq:asterisk X}. We use a partial transpose with respect to the first subsystem and note for each term that $(K_n(\trans)|\psi^i\rangle\langle n|)^\top=K_n(1-\trans)|\psi^i\rangle\langle n|$ as for pure states:
% \begin{equation}
%     \begin{aligned}
%         \Tilde{X}_\trans^{\top_a}&=\sum_i\sqrt{p_i} \sum_n K_n(1-\trans)|\psi^i\rangle\langle n|\otimes|0\rangle\langle \psi^i|\\
%     &=\tilde{X}_{1-\trans}.
%     \end{aligned}
% \end{equation} 
% Since the full transposition is the composition of two partial transpositions,
% \begin{equation}
%     \begin{aligned}
%         (\Tilde{X}_\trans^\dagger &\Tilde{X}_\trans)^*=\Tilde{X}_\trans^{\top_a\top_b} \Tilde{X}_\trans^{\top_a\top_b\dagger}\\
%         &=\Tilde{X}_{1-\trans}^{\top_b} \Tilde{X}_{1-\trans}^{\top_b\dagger}\\
%         &=\sum_m K_m(1-\trans)\otimes\openone|\Psi\rangle\langle m,0|\\
%         &\quad\times
%         \sum_n |n,0\rangle\langle \Psi|K_n(1-\trans)^\dagger\otimes\openone\\
%         &=\sum_n K_n(1-\trans) \otimes \openone|\Psi\rangle\langle \Psi|K_n(1-\trans)^\dagger \otimes \openone,
%     \end{aligned}
% \end{equation}which is Eq.~\eqref{eq:asterisk X}.

    \section{Details leading to Theorem 3}

We now provide proof of a lemma that leads to Theorem 3 in the main text. This establishes the overlap between two operators, each subject to loss, as a polynomial in $1-2\trans$ with positive coefficients that depend on the two operators. The proof relies on an experimental method for measuring purity, the fact that uniform losses commute with linear optics, and that known loss prior to photon-number measurement can be incorporated into the statistics of the measurement using post processing.
\begin{lem}
    [Positive polynomial expansion of Hilbert-Schmidt norm]\label{lem:HS positive} The overlap between any two positive operators {$\op_1$} and {$\opp_1
    $} subject to loss, as measured by the Hilbert-Schmidt inner product {$O_\trans=\tr[\op_\trans \opp_\trans]$}, is a polynomial in $\cent\equiv (1-2\trans)$ whose expansion coefficients are all nonnegative.
\end{lem}
\begin{proof}
    Rewrite the inner product into a two-mode version using the swap operator $S=(-1)^{n_-}$ for $n_i=a^\dagger_i a_i$ and $a_\pm=(a_1\pm a_2)/\sqrt{2}$ as {$O_\trans=\tr[\op_\trans \otimes \opp_\trans S]$} (this method has been used to measure purity using two copies of a state~\cite{Bovinoetal2005,Islametal2015}). Treating loss as beam splitters to vacuum modes annihilated by $b_1$ and $b_2$, noting that $|0\rangle_{b_1}\!\otimes \!|0\rangle_{b_2}\!=\!|0\rangle_{b_+}\!\otimes |0\rangle_{b_-}$ and $B_{a_1b_1}(T)B_{a_2b_2}(T)\!=\!B_{a_+b_+}(T)B_{a_-b_-}(T)$, and using cyclic permutations inside the trace to set $\langle 0|B_{a_+ b_+}^\dagger B_{a_+b_+}|0\rangle_{b_+}=\openone_{a_+}$:
  %  \eq{
    % \tr[&B_{a_1b_1}B_{a_2b_2}\rho\otimes\sigma\otimes \ketbra{0,0}{0,0}_{b_1,b_2}B_{a_1b_1}^\dag B_{a_2b_2}^\dag\ (-1)^{\hat{n}_{-}}]\\
    % &=\tr[ \rho\otimes\sigma \langle 0|B_{a_-b_-}^\dag\ (-1)^{\hat{n}_{-}} B_{a_-b_-}|0\rangle_{b_-}]
    % }
    \begin{equation}
%         \begin{aligned}
%             \tr[&B_{a_1b_1}B_{a_2b_2}{\rho_1}\otimes{\sigma_1}\otimes \ketbra{0,0}{0,0}_{b_1,b_2}B_{a_1b_1}^\dag B_{a_2b_2}^\dag\ (-1)^{\hat{n}_{-}}]\\
%     &=\tr[{\rho_1}\otimes{\sigma_1} \langle 0|B_{a_-b_-}^\dag\ (-1)^{\hat{n}_{-}} B_{a_-b_-}|0\rangle_{b_-}]
% \nonumber 
%         \end{aligned}
\begin{aligned}
            \tr[&B_{a_1b_1}B_{a_2b_2}{\op_1}\otimes{\opp_1}\otimes \ketbra{0,0}{0,0}_{b_1,b_2}B_{a_1b_1}^\dag B_{a_2b_2}^\dag\ (-1)^{\hat{n}_{-}}]\\
    &=\tr[{\op_1}\otimes{\opp_1} \langle 0|B_{a_-b_-}^\dag\ (-1)^{\hat{n}_{-}} B_{a_-b_-}|0\rangle_{b_-}].
%\nonumber 
        \end{aligned}
    \end{equation} Then resolving the identity over states $|n\rangle_{b_-}$ and noting the Kraus operators for loss to be $$\langle n|B_{a_-b_-}|0\rangle_{b_-} =\sqrt{\trans}^{n_-}(\sqrt{1-\trans}a_- )^n/\sqrt{n!}$$ as in Lemma~\ref{lem:T vs 1-T}, we sum $$\sum_{n\geq 0}(1-\trans)^n a^{\dagger n} (-\trans)^{a^\dagger a}a^n/n!=\sum_{m\geq 0}(1\!-\!2T)^m |m\rangle\langle m|$$ and recognize $B_{a_1 a_2}(\tfrac{1}{2})a_2 B_{a_1 a_2}(\tfrac{1}{2})^\dagger=a_-$ to find 
    \begin{equation}
        O_\trans\!=\!\sum_{m\geq 0}\!\lambda^m \tr[({\op_1}\otimes{\opp_1}) B_{a_1a_2}(\tfrac{1}{2})(\openone \!\otimes\!|m\rangle\langle m|)B_{a_1 a_2}(\tfrac{1}{2})^\dagger].\nonumber
    \end{equation} Each term $B_{a_1a_2}(\tfrac{1}{2})(\openone \!\otimes\!|m\rangle\langle m|)B_{a_1 a_2}(\tfrac{1}{2})^\dagger$ is a projection operator and thus its trace against a positive operator is always positive.
\end{proof}

% \setcounter{thm}{1}
% \begin{thm}[Concavity of mixedness/convexity of purity]\label{thm:pur}
%     The purity of a pure state subject to loss $\pur_\psi(\trans)$ is an everywhere convex function and likewise the mixedness $\mix_\psi(\trans)$ is everywhere concave; $\partial ^2 \pur_\psi(\trans)/\partial \trans^2\geq 0\, \forall \trans\in(0,1)\, \forall \psi$.
% \end{thm}

% \setcounter{thm}{\getrefnumber{lem:HS positive}}

% \begin{proof}
%     By Lemma \ref{lem:HS positive}, $\pur_\psi(\trans)=\sum_{m\geq 0}\lambda^m p_m$ for coefficients $p_m\geq 0$. By symmetry about $\trans\!=\!1/2$ ($\lambda=0$), $p_m=0$ for all odd $m$. The second derivative is therefore comprised solely of positive terms $\frac{\partial^2 \pur}{\partial \trans^2}=\sum_{l> 0}4(2l)(2l-1)\lambda^{2l-2}p_{2l} >0$.
% \end{proof}

\section{Proof of Theorem 4}

{We explain here how to compute the product of the Schmidt coefficients of a finite-dimensional pure state subject to loss in order to prove our theorem that $G$-concurrence is log-concave and its corresponding entanglement monotone is maximized by balanced beam splitters.}
 \setcounter{thm}{3}
\begin{thm}[Log-concavity of concurrence]
\label{thm:conc}
    The logarithm of the $G$-concurrence of a finite-dimensional pure state subject to loss 
    is a concave function of loss; 
    $\forall\,  |\{\psi_n\}_n|<\infty, \trans \in [0, 1]$: $\partial^2  \log\conc/\partial \trans^2\leq~0$. 
    Any pure state interfered with the vacuum at a beam splitter thus generates the greatest $G$-concurrence entanglement (entanglement capacity~\cite{Gour2005}) when the beam splitter has transmission $\trans\!=\!1/2$.
\end{thm}
\setcounter{thm}{\getrefnumber{lem:HS positive}}

\begin{proof}
    For finite-dimensional states, there is a maximal $N$ such that $\psi_n=0$ for all $n>N$. Then the Schmidt matrix $M$ is an anti-triangular matrix with zeros below the main anti-diagonal.
    The singular values $\{s_0,\cdots,s_N\}$ of $M$ are complicated, but their product can be computed and is the only property needed in  $\det[\rho_\trans]=\prod_{i=0}^N s_i^2=\det{(MM^\dagger)}=|\det M|^2$. As an anti-triangular matrix, the determinant of $M$ is given by the product of the entries on the main anti-diagonal, all multiplied by $(-1)^{N}(-1)^{N-1}\cdots$. Together,
    %\eq{
    %\conc_\psi(\trans)=|\psi_N|^{N+1}\frac{N!^{\tfrac{N+1}{2}}}{\prod_{n=0}^N n!} [\trans(1-\trans)]^{\tfrac{N(N+1)}{4}}.
    %.
    %} 
    with $d=N+1$
    \begin{equation}
        \begin{aligned}
            \conc_\psi(\trans)&=(N+1)\left|\prod_{m=0}^{N} M_{m,N-m}\right|^{\tfrac{2}{N+1}}\\&=(N+1)\left|\prod_{m=0}^{N} \psi_N\sqrt{\binom{N}{m}\trans^m (1-\trans)^{N-m}}\right|^{\tfrac{2}{N+1}}\\
            &=(N+1)|\psi_N|^2 \prod_{m=0}^N \binom{N}{m}^{\tfrac{1}{N+1}}\\
            &\quad\times\left(\trans^{\sum_{m=0}^N m} (1-\trans)^{\sum_{m=0}^N (N-m)}\right)^{\tfrac{1}{N+1}}
            \\
            %&=(N+1)|\psi_N|^2\frac{{N!}}{\prod_{n=0}^N n!^{\tfrac{2}{N+1}}} [\trans(1-\trans)]^{\tfrac{N}{2}}.\\
            &=(N+1)|\psi_N|^2[\trans(1-\trans)]^{\tfrac{N}{2}} \left(C_N\right)^{\tfrac{1}{N+1}}.
        \end{aligned}
    \end{equation}
The combinatorial integer $C_N = \prod_{m=0}^N \binom{N}{m}$ can also be expressed in terms of the Barnes $G$ function via $C_N = N!^N/(G(N+1) G(N+2)$ where $G(N) = \prod_{m=0}^{N-2} m!$. Similar to the concavity of the matrix function $\log\det A$, we see that 
\begin{equation}
\log\conc=\frac{N}{2}\log \left (\trans(1-\trans) \right ) + \log(|\psi_N|^2 (N+1)) + \frac{\log(C_N)}{N+1}.
\label{eq:conc log convex}
\end{equation}
This is manifestly concave with global maximum at $\trans\!=\!1/2$; the proof is completed using monotonicity of a logarithm with its argument. 
    % The functional dependence on $\trans$ is a positive power of $\trans(1-\trans)$, which increases with $\trans$ for $\trans\leq 1/2$, so by symmetry there is a global maximum at $\trans\!=\!1/2$.
    The reason that only $\log \conc$ is concave while $\conc$ itself is only monotonic on either side of $\trans\!=\!1/2$ is because the latter is a positive power of a concave function $\trans(1-\trans)$, which does not guarantee concavity.
\end{proof}

{One could consider extending this proof to infinite-dimensional pure states. However, one encounters difficulties in convergence, as explained in the next section.}
\section{Convergence of determinants in infinite dimensions}
\begin{lem}[Entanglement convergence in infinite dimensions]\label{lem:inf}
    The Schmidt coefficients of an infinite-dimensional pure state with finite energy are equal to those found from a limit of finite-dimensional pure states. 
\end{lem}
\begin{proof}
The squares of the Schmidt coefficients are the eigenvalues of the positive-definite matrix $M M^\dagger$.
    For a finite-dimensional state with maximal Fock number $N$, denote the matrix from which the Schmidt coefficients are found as $M_N$. We prove that the eigenvalues $\{\lambda_i\}$ of $M_N$ converge to those of $M$ as $M_N\to M$. First, note that the nonzero eigenvalues are the solutions of the equation $\det(\mathds{I}-A/\lambda)=0$ and that the Fredholm determinant $\det(\mathds{I}+A)$ is well-defined and continuous on trace-class operators $A$. 

    Next, defining the differences between the Fredholm determinants as $\delta =\left|\det(\mathds{1}-M/\lambda)-\det(\mathds{1}-M_N/\lambda)\right|$, the Fredhold determinant obeys the inequality
    \eq{
    \delta \leq ||M-M_N||_1 \exp[\max(||M_N||_1,||M||_1)/|\lambda|+1]/|\lambda|.
    } 

    The eigenvalues of $M$ and $M_N$ approach each other. This can be seen by verifying the trace $\tr[M]=\sum_{n=0}^\infty |\psi_n| \sqrt {p_{n/2}(n;\trans)}<\infty$ and the norms $\sum_{m\geq 0 ,n> N}|\psi_n|\sqrt{p_m(n;\trans)}=||M-M_N||_1<||M||_1=\sum_{mn\geq 0}|\psi_n|\sqrt{p_m(n;\trans)}<\infty$, where we use the binomial probability $p_m(n;\trans)=\binom{n}{m}\trans^{m}(1-\trans)^{n-m}$. When $n$ is large, the binomial distribution tends to a normal distribution with mean $\mu=n\trans$ and variance ${\sigma^2=n\trans(1-\trans)}$. The trace of $M$ sums over the peaks of the square roots of this distribution, with values $p_{\mathrm{max}}(n;\trans)\approx (\sqrt{2\pi}\sigma)^{-1}=1/\sqrt{2\pi \trans(1-\trans)n}$. Assume the state has finite energy $\infty >\langle a^\dagger a\rangle=\sum_n |\psi_n|^2 n$; then $|\psi_n|< \mathcal{O}(1/n)$. Then $\tr(M)< \sum_n \mathcal{O}(n^{-3/2})<\infty$ (this is also true under the weaker condition that the state's coefficients are normalizable, with $\sum_n |\psi_n|^2 <\infty \Rightarrow |\psi_n|<\mathcal{O}(1/\sqrt{n})$).

    Next, the one norms of the matrices are computed as the finite one norms of the $M_N$ plus the one norm of the extra components $||M||_1=||M_N||_1+||M-M_N||_1$. The latter, for sufficiently large $N$, can make use of the square root of a normal distribution, summed over all components. The square root of a normal distribution is also normal, with the same mean, but with a larger standard deviation $\sigma^2=2n\trans(1-\trans)$ and the whole distribution is multiplied by $\sqrt{\sqrt{2\pi n \trans(1-\trans)}}\sqrt{2}$. The sum over $p_m$ becomes an integral over the normal distribution, yielding the coefficient $\sim n^{1/4}$. We thus find $||M-M_N||_1=\sum_{n>N}|\psi_n|\mathcal{O}(n^{1/4})$. Choosing states with finite higher-order moments of energy, such as $\langle (a^\dagger a)^{3/2}\rangle <\infty \Rightarrow |\psi_n|<\mathcal{O}(n^{-5/2})$, ensures that $||M-M_N||_1<\sum_{n>N}\mathcal{O}(n^{1/4-5/4})<\infty$. Once these are all proven to be finite, the work is done: the norm decreases as $N$ gets bigger, one can make it as small as any parameter $\epsilon \to 0$, thus sending $\delta\to 0$.

    Since the eigenvalues of $M$ approach those of $M_N$, the products of the eigenvalues of $M_N$ approach those of $M$. Then, so too does the product of Schmidt coefficients of the state and thus the finite-dimensional entanglement monotone given by the product of Schmidt coefficients seems to converge in the limit of infinite dimensions. However, the entanglement monotone for $N$ is given by $\prod_{i=0}^N s_i$ and for $N+1$ by $s_{N+1}\prod_{i=0}^N s_i$. While the latter is continuous with $s_{N+1}$, there is a discontinuity between the former and the latter, especially noticeable when $s_{N+1}\approx 0$. Should the entanglement monotone have been defined as $\prod_{i=0}^\infty s_i$ for all states, it would have always been continuous, but it would have been vanishingly small for all states. The entanglement monotone defined with determinants of reduced density matrices is thus only applicable for finite-dimensional states.
\end{proof}

\section{R\'enyi entropies of Fock-state inputs}

The main text showed that the transmission parameter $\trans$ that maximizes the entanglement generated by a Fock state $|N\rangle$ does not always occur at $\trans\!=\!1/2$ when entanglement is measured by a R\'enyi entropy with sufficiently large index $\alpha$. There, it appeared that $\alpha=6$ was special for the state with $N=6$. We record in Table~\ref{tab:max renyi} the maximal integer value of $\alpha$ above which $\trans\!=\!1/2$ no longer optimizes the R\'enyi-$\alpha$ entropy of entanglement. We see that this value of $\alpha$ can be both above and below $N$, showcasing the richness of the problem when already restricted to Fock states; it is surely even more dramatic when all possible input states are considered.

\begin{table}[h!]
\caption{Largest integer R\'enyi entropy order such that an input Fock state maximizes entropy after losing $50\%$ of its photons and maximizes two-mode entanglement as measured by $H_\alpha$ at $\trans=1/2$.}
\label{tab:max renyi}
\begin{tabular}{cc}
Fock state number $N$ & Maximal R\'enyi index $\alpha$ \\ \hline
1 & $\infty$ \\
2 & 3 \\
3 & $\infty$ \\
4 & 5 \\
5 & $\infty$ \\
6 & 6 \\
7 & $\infty$ \\
8 & 8 \\
9 & $\infty$ \\
10 & 9 \\
11 & $\infty$ \\
12 & 11 \\
13 & $\infty$ \\
14 & 12 \\
15 & $\infty$ \\
16 & 13 \\
17 & $\infty$ \\
18 & 15 \\
19 & $\infty$ \\
20 & 16 
\end{tabular}
\end{table}

\end{document}